\documentclass{llncs}

\usepackage{amsmath}
\usepackage{amssymb}
\usepackage{cancel}
\usepackage{allrunes}
\usepackage{graphics}
\usepackage{verbatim}
\usepackage{color}
\usepackage{tikz}
\usepackage[curve]{xypic}
\usetikzlibrary{arrows}
\usepackage{algorithm2e}
\usepackage{fontenc}

\graphicspath{{./img/}}

\newcommand{\restr}[2]{{#1}|_{#2}}

\newcommand{\agents}{\mathcal A}
\newcommand{\lang}{\mathcal L}

\newcommand{\lblack}{\lang^{\cdia}}

\newcommand{\ddia}[1]{\langle{#1}\rangle}
\newcommand{\dbox}[1]{[#1]}
\newcommand{\adia}{\lozenge}

\newcommand{\abox}{\square}
\newcommand{\cdia}{\blacklozenge}
\newcommand{\cbox}{\blacksquare}

\newcommand{\acro}[1]{\textsc{#1}}

\newcommand{\iterate}[1]{{\cal I}(#1)}

\newcommand{\shrink}[2]{\rho_{#1}(#2)}
\newcommand{\comp}[2]{C(#1,#2)}

\newcommand{\dlangm}{{\mathcal L}_{\acro{DDL}}}
\newcommand{\pis}[1]{{\mathbf e}_{#1}}
\newcommand{\carriers}[1]{Q_{#1}}
\newcommand{\kmod}[2]{{\cal K}_{(#1,#2)}}
\newcommand{\rels}[1]{{\sf R_{#1}}}
\newcommand{\update}[3]{{\mathcal U}_{#1}(#2,#3)}
\newcommand{\cons}[1]{{\textara{\ea}}(#1)}
\newcommand{\af}{{\sf F}}

\newcommand{\basis}{basis }

\newcommand{\views}{\mathcal B}
\newcommand{\viewsv}{\left(V_a\right)_{(a \in \agents)}}
\newcommand{\carrier}{Q_\views}
\newcommand{\sem}{\varepsilon}
\newcommand{\depth}[1]{|{#1}|^\adia}
\newcommand{\bisim}{\underline{\leftrightarrow}}

\author{Truls Pedersen\inst{1} \and Sjur Dyrkolbotn\inst{2}}
\institute{Inst. for Information Science and Media Studies, University of Bergen, Norway \and Durham Law School, Durham University, United Kingdom}
\title{Computing consensus: A logic for reasoning about deliberative processes based on argumentation}

\begin{document}

\maketitle

\begin{abstract}
We consider multi-agent argumentation, where each agent's view of the arguments is encoded as an argumentation framework (AF). Then we study deliberative processes than can occur on this basis. We think of a deliberative process as taking the shape of a stepwise aggregation of a single joint AF, and we are interested in reasoning about the space of possible outcomes. The only restriction we place on deliberative processes is that they should satisfy \emph{faithfulness}, a postulate amounting to requiring that whenever deliberation leads to a new relationship being introduced between two arguments, this relationship is endorsed by at least one participating agent. We use modal logic to reason about the resulting deliberative structures, and we provide some technical results on model checking. We also give an example and suggest some directions for future work.
\end{abstract}

\section{Introduction}\label{sec:intro}

We study agency and argumentation, and propose a framework for modeling and reasoning about \emph{deliberation}. In short, we assume that some agents are given, with their own individual view of the argumentation scenario at hand, and that deliberation is a process by which we attempt to reconcile all these individual views to aggregate a joint, common view on the situation, which we may then analyze further using established techniques from argumentation theory. We develop our framework based on the formal theory of argumentation introduced by Dung \cite{Dung}, which has attracted much interested from the AI community, see \cite{Argbook} for a recent volume devoted to this theory. In keeping with recent trends, we also take advantage of \emph{logical} tools, relying both on a truth-functional three-valued view of argumentation \cite{Arieli,Sjur-SYNT}, and on the use of modal logic \cite{grossi1,grossi2}.

Unlike much previous work in argumentation theory, including work done on multi-agent argumentative interaction (see, for instance, \cite[Chapter 13]{Argbook}), we do not worry about attempting to design procedures for ``good'' deliberation, but focus instead on a logical analysis of the space of possible outcomes, assuming only a minimal restriction on the nature of the  deliberative processes that we consider permissible. The restriction encodes the intuition that the common view aggregated by deliberation should be built in such a way that we only make use of information already present in the view of some agent. This, indeed, seems like a safe requirement, and appears to be one that no reasonable group of agents would ever want to deviate from. 

We remark that our concept of deliberation is also somewhat unusual compared to earlier work in that we focus on processes that aim to reconcile \emph{representations} of the argumentative situation, rather than processes where agents interact based on their different \emph{judgments} about a given structure. While we will not argue extensively for the appropriateness of this shift of attention, we mention that it is prima facie reasonable: if two agents judge a situation differently it seems safe to assume that they must have a \emph{reason} for doing so. Moreover, while this reason could sometimes be due to disagreement about purely logical principles or atomic (i.e., unanalyzable) differences in preference, it seems clear that in real life, disagreement arises just as often, perhaps more often, as a result of a difference in \emph{interpretation}, i.e., from the fact that the agents have different mental representations of the situation at hand.

In fact, in this paper we will refrain from committing to a particular view on judgments, and we will set up our logical framework in such a way that it allows us to use any semantics from formal argumentation as the source of basic judgments about arguments given a framework. Our own focus is solely on the stepwise, iterative development of a common framework, and on the logical analysis of the different ways in which such a process may unfold, by way of a logical treatment of the modalities that arise from quantifying over the space of possible deliberative futures.

The structure of the paper is as follows. In Section \ref{sec:abt} we give a background on abstract argumentation, concluding by an example that motivates our logic and a definition of \emph{normality}, which allows us to parameterize our constructions by any normal argumentation semantics (all semantics of which we are aware are normal). Then in Section \ref{sec:ddl} we introduce \emph{deliberative dynamic logic} (\acro{ddl}), a concrete suggestion for a logical framework allowing us to reason about deliberative processes using modal logic. Then in Section \ref{sec:mcheck} we show that while our models are generally infinite, model checking is still feasible since we may ``shrink'' them, by restricting attention only to the relevant parts of the structure. Then, in Section \ref{sec:concfut}, we conclude and discuss directions for future work.

\section{Background on abstract argumentation}\label{sec:abt}

Following Dung \cite{Dung} we represent argumentation structures by directed graphs $\af = (S,E)$, such that $S$ is a set of arguments and $E \subseteq S \times S$ encode the attacks between them, i.e., such that if $(x,y) \in E$ then the argument $x$ attacks $y$. Traditionally, most work in formal argumentation theory has focused on defining and investigating notions of successful sets of arguments, in a setting where the argumentation framework is given and remains fixed. Such notions are typically formalized by an \emph{argumentation semantics}, an operator $\sem$ which returns, for any AF $\af$, the set of sets of arguments from $\af = (S,E)$ that are regarded as successful combinations, i.e., such that $\sem(\af) \subseteq 2^S$. Many proposals exists in the literature, we point to \cite{BaroniEval} for a survey and formal comparison of different semantics. While some semantics, such as the grounded and ideal semantics, return a unique set of arguments, the ``winners'' of the argumentation scenario encoded by $\af$, most semantics return more than one possible collection of arguments that \emph{would} be successful if they were held together. For instance, the admissible semantics, upon which many of the other well-known semantics are built, returns, for each AF $\af$, the following sets of arguments:

$$
a(\af) = \{A \subseteq S \mid E^-(A) \subseteq E^+(A) \subseteq S \setminus A\}
$$

Where $E^-(X)$ are the nodes which has an arrow into $X$, and $E^+(X)$ are the nodes which the nodes in $X$ has an arrow into. That is, the admissible sets are those that can defend themselves against attacks (first inclusion), and do not involve any internal conflicts (second inclusion). A strengthening that is widely considered more appropriate (yet incurs some computational costs) is the \emph{preferred} semantics $\sf p$, which is defined by taking only those admissible sets that are set-theoretically maximal, i.e., such that they are not contained in any other admissible set. In general, an AF admits many preferred sets, and even more admissible ones. Indeed, notice that the empty set is always admissible by the default (the inappropriateness of which provides partial justification for using preferred semantics instead). As a simple example, consider $\af$ below.
$$\begin{array}{ll}
\af: \xymatrix{p \ar@/_/[r] & q \ar@/_/[l] } &\hspace{4em} a(\af) = \{\emptyset,\{p\},\{q\}\}, \ \sf p(\af) = \{\{p\},\{q\}\}\}
\end{array}
$$

Indeed, it seems hard to say which one of $p$ and $q$ should be regarded as successful in such a scenario. In the absence of any additional information, it seems safest to concede that choosing either one will be a viable option. Alternatively, one may take the view that due to the undetermined nature of the scenario, it should not be permitted to regard either argument as truly successful. This, indeed, is the view taken by unique status semantics, such as the grounded and ideal semantics. However, while such a restrictive view might be appropriate in some circumstance, it seems unsatisfactory for a general theory of argumentation. Surely, in most real-world argumentation situations, it is not tenable for an arbitrator to refrain from making a judgment whenever doing so would involve some degree of discretion on his part.

Since argumentation semantics typically only restrict the choice of successful arguments, without determining it completely, a modal notion of \emph{acceptance} arises, usually referred to as \emph{skeptical} acceptance in argumentation parlance, whereby an argument is said to be skeptically accepted by $\af$ under $\sem$ if $\forall S \in \sem(\af): p\in S$. The dual notion is called \emph{credulous} acceptance, and obtains just in case $\exists S \in \sem(\af): p\in S$. Moreover, since the choice among elements of $\sem (\af)$ can itself be a contentious issue, and is not one which can be satisfactorily resolved by single-agent argumentation theory, there has been research devoted to giving an account of multi-agent interaction concerning the choice among members of $\sem(\af)$, see \cite{manipulation,lying}. While this is interesting, it seems that another aspect of real-world argumentation has an even stronger multi-agent flavor, namely the process by which one arrives at a common AF in the first place. Certainly, two agents, $a$ and $b$, might disagree about whether to choose $p$ or $q$ in $\af$ considered above, but as it stands, such a choice appears arbitrary and, most likely, the two agents would also be willing to admit as much. Arguably, then, the disagreement itself is only superficial. The agents disagree, but they provide no \emph{reason} for their different preferences, and do not provide any content or structure to substantiate them. This leaves an arbitrator in much the same position as he was in before: he might note the different opinions raised, but he has no basis upon which to inquire into their merits, and so his choice must, eventually, still be an exercise in discretion.

In practice, however, it would have to be expected that if the agents $a$ and $b$ were really committed to their stance, they would not simply accept that $\af$ correctly encodes the situation and that the choice is in fact arbitrary. Rather, they would produce \emph{arguments} to back up their position. It might be, for instance, that agent $a$, who favors $p$, claims that $q$ is inconsistent for some reason, while agent $b$, who favors $q$, makes the same accusation against the argument $p$. Then, however, we are no longer justified in seeing this as disagreement about which choice to make from $\sem(\af)$. Rather, the disagreement concerns the nature of the argumentation structure itself. The two agents, in particular, put forth different \emph{views} on the situation. For instance, in our toy example, we would have to consider the following two AFs, where $V_a, V_b$ encode the views of $a$ and $b$ respectively.

\begin{equation}\label{ex}
\begin{array}{ll}
V_a: \hspace{2em} \xymatrix{p \ar@(lu,ld) \ar@/_/[r] & q \ar@/_/[l] } & \hspace{4em} V_b: \xymatrix{p \ar@/_/[r] & q \ar@(ru,rd) \ar@/_/[l] }
\end{array}
\end{equation}

Then the question arises: what are we to make of this? 

In the following, we address this question, and we approach it from the conceptual starting point that evaluating (higher-order) differences of opinion such as that expressed by $V_a,V_b$ takes place iteratively, through a process of \emph{deliberation}, leading, in a step-by-step fashion, to an aggregated \emph{common} $\af$. Such a process might be instantiated in various ways: it could the agents debating the matter among themselves and reaching some joint decision, or it could be an arbitrator who considers the different views and reasons about them by emulating such a process. Either way, we are not interested in attempting to provide any guidance towards the ``correct'' outcome, which is hardly possible in general. Rather, we are interested in investigating the modalities that arise when we consider the space of all possible outcomes (where possible will be defined in due course). Moreover, we are interested in investigating structural questions, asking, for instance, about the importance of the order in which arguments are considered, and the consequences of limiting attention to only a subset of arguments.

We use a dynamic modal logic to facilitate this investigation, and in the next section we define the basic framework and show that model checking is decidable even on infinite AF's, as long as the agent's views remains finitely branching, i.e., as long as no argument is attacked by infinitely many other arguments. We will parameterize our logic by an argumentation semantics, so that it can be applied to any such semantics which satisfies a normality condition. In particular, let $C(\af) = \{C^{\af}_1,\ldots,C^{\af}_i,\ldots\}$ denote the (possibly infinite) set of maximal connected components from $\af$ (the set of all maximal subsets of $S$ such that any two arguments in the same set are connected by a sequence of attacks). Then we say that a semantics $\sem$ is \emph{normal} if we have, for any $\af = (S,E)$

\begin{equation}\label{eq:normal}
A \in \sem(\af) \Leftrightarrow A = \bigcup_{i}A_i \text{ for some } A_1,\ldots,A_i,\ldots \text{ s.t. } A_i \in \sem(C^{\af}_i) \text { for all } i
\end{equation}
That is, a semantics is normal if the status of an argument depends only on those arguments to which it has some (indirect) relationship through a sequence of attacks. We remark that all argumentation semantics of which we are aware satisfies this requirement, hence we feel justified in dubbing it normality.

\section{Deliberative dynamic logic}\label{sec:ddl}

We assume given a finite non-empty set $\agents$ of agents and a countably infinite set $\Pi$ of arguments.\footnote{Possibly ``statements'' or ``positions'', depending on the context of application.} The basic building block of dynamic deliberative logic is provided in the following definition.

\begin{definition}\label{def:basis} A \basis for deliberation is an $\agents$-indexed collection of digraphs $\views = (V_a)_{(a \in \agents)}$, such that for each $a \in \agents$, $V_a \subseteq \Pi \times \Pi$.
\end{definition}

Given a basis which encodes agents' view of the arguments, we are interested in the possible ways in which agents can deliberate to reach \emph{agreement} on how arguments are related. That is, we are interested in the set of all AFs that can plausibly be seen as resulting from a \emph{consensus} regarding the status of the arguments in $\Pi$. What restrictions is it reasonable to place on a consensus? It seems that while many restrictions might arise from pragmatic considerations, and be implemented by specific protocols for ``good'' deliberation in specific contexts, there are few restrictions that can be regarded as completely general. For instance, while there is often good reason to think that the position held by the majority will be part of a consensus, it is hardly possible to stipulate an axiomatic restriction on the notion of consensus amounting to the principle of majority rule. Indeed, sometimes deliberation takes place and leads to a single dissenting voice convincing all the others, and often, these deliberative processes are far more interesting than those that transpire along more conventional lines. However, it seems reasonable to assume that whenever \emph{all} agents agree on how an argument $p$ is related to an argument $q$, then this relationship is part of any consensus. This, indeed, is the only restriction we will place on the notion of a consensus; that when the AF $\af$ is a consensus for $\basis$, it must satisfy the following \emph{faithfulness} requirement.
\begin{itemize}
\item \emph{For all $p,q \in \Pi$, if there is no disagreement about $p$'s relationship to $q$ (attack/not attack), then this relationship is part of $\af$}
\end{itemize}

This leads to the following definition of the set $\cons \views$, which we will call the set of \emph{complete assents} for $\views$, collecting all AFs that are faithful to $\views = \viewsv$.

\begin{equation}\label{def:consensus}
\cons \views = \left\{\af \subseteq \Pi \times \Pi ~\left|~ \bigcap_{a \in \agents}V_a \subseteq \af \bigcup_{a \in \agents}V_a\right.\right\}
\end{equation}

An element of $\cons \views$ represents a possible consensus among agents in $\agents$, but it is an \emph{idealization} of the notion of assent, since it disregards the fact that in practice, assent tends to be \emph{partial}, since it results from a dynamic process, emerging through \emph{deliberation}. Indeed, as long as the number of arguments is not bounded we can \emph{never} hope to arrive at complete assent via deliberation. We can, however, initiate a process by which we reach agreement on more and more arguments, in the hope that this will approximate some complete assent, or maybe even be \emph{robust}, in the sense that there is \emph{no} deliberative future where the results of current partial agreement end up being undermined. Complete assent, however, arises only in the limit.

When and how deliberation might successfully lead to an approximation of complete assent is a question well suited to investigation with the help of dynamic logic. The dynamic element will be encoded using a notion of a deliberative event -- centered on an argument -- such that the set of ways in which to relate this argument to arguments previously considered gives rise to a space of possible deliberative time-lines, each encoding the continued stepwise construction of a joint point of view. This, in turn, will be encoded as a monotonically growing AF $\af = (S,E)$ where $S \subseteq \Pi, E \subseteq S \times S$ and such that faithfulness is observed by all deliberative events. That is, an event consists of adding to $\af$ the agents' combined view of $p$ with respect to the set $S \cup \{p\}$. This leads to the following collection of possible events, given a basis $\views$, a partial consensus\footnote{These ``partial consensuses'' are sometimes referred to as ``contexts'' when they are used to describe graphs inductively, as we will do later.} $\af = (S,E)$ and an argument $p \in \Pi$:

\begin{equation}\label{eq:update}
\update \views \af p = \left\{X ~\left|~ \bigcap_{a \in \agents}\restr {V_a} {S \cup \{p\}} \subseteq X \subseteq \bigcup_{a \in \agents}\restr {V_a}{S \cup \{p\}}\right.\right\}
\end{equation}

To provide a semantics for a logical approach to deliberation based on such events, we will use Kripke models.

\begin{definition}[Deliberative Kripke model]\label{def:main} Given an argumentation semantics $\sem$ and a set of views $\views$, the deliberative Kripke models induced by $\views$ and $\sem$ is the triple $\kmod \views \sem = (\carriers \views, \rels \views,\pis \sem)$ such that
\begin{itemize}
\item $\carriers \views$, the set of points, is the set of all pairs of the form $q = (q_S,q_E)$ where $q_S \subseteq \Pi$ and $$\bigcap_{a \in \agents}\restr {V_a} {q_S} \subseteq q_E \subseteq \bigcup_{a \in \agents}\restr {V_a} {q_S}$$
The basis $\views$ together with our definition of an event, given in Equation \ref{eq:update}, induces the following function, mapping states to their possible deliberative successors, defined for all $p \in \Pi, q \in \carriers \views$ as follows
$$succ(p, q) \quad := \quad \{~(q_S \cup \{p\}, q_E \cup X) ~|~ X \in \update \views q p~\}$$ 
We also define a lifting, for all states $q \in \carriers \views$:
$$succ(q) \quad := \quad \{~q' \mid \exists p \in \Pi: q' \in succ(q,p)\}$$
\item $\rels \views: \Pi \cup \{\exists\} \to 2^{\carriers \views \times \carriers \views}$ is a map from symbols to relations on $\carriers \views$ such that 
\begin{itemize} \item $\rels \views(p) = \{(q,q') \mid q' \in succ(p,q)\}$ for all $p \in \Pi$ and
\item $\rels \views(\exists) = \{(q,q') \mid q' \in succ(q)\}$,
\end{itemize} \vspace{1em}
\item $\pis \sem: \carriers \views \to 2^{(3^\Pi)}$ maps states to labellings such that for all $q \in \carriers \views$ we have $\pis \sem(q) = (\pi_1,\pi_0,\pi_{\frac{1}{2}})$ with $$\pis \sem(q) = \{\pi \mid \pi_1 \in \sem(q), \pi_0 = \{p \in q_S \mid \exists q \in \pi_1: (q,p) \in q_E\}\}$$
\end{itemize}
\end{definition}

Notice that in the last point, we essentially map $q$ to the sets of extensions prescribed by $\sem$ when $q$ is viewed as an AF. We encode this extension as a three-valued labeling, however, following \cite{caminada06}. Notice that the default status, attributed to all arguments not in $q_S$, is $\frac{1}{2}$. The logical language we will use consists in two levels. For the lower level, used to talk about static argumentation, we follow \cite{Arieli,Sjur-SYNT} in using {\L}ukasiewicz three-valued logic. Then, for the next level, we use a dynamic modal language which allows us to express consequences of updating with a given argument, and also provides us with existential quantification over arguments, allowing us to express claims like ``there is an update such that $\phi$''. This leads to the language $\dlangm$ defined by the following \acro{bnf}'s

$$ \phi \quad ::= \cdia \alpha ~|~ \neg \phi ~|~ \phi \wedge \phi ~|~ \ddia p \phi ~|~ \adia \phi$$ 
where $p \in \Pi$ and $\alpha \in \lblack$ where $\lblack$ is defined by the following grammar:
$$
\alpha ::= p \ | \ \neg \phi \ | \ \phi \to \phi $$
for $p \in \Pi$.

We also use standard abbreviations such that $\abox \phi = \neg \adia \neg \phi$, $\dbox p \phi = \neg \ddia p \neg \phi$ and $\cbox \alpha = \neg \cdia \neg \alpha$. We also consider that standard boolean connectives abbreviated as usual for connectives not occurring inside a $\cdia$-connective and abbreviations for connectives of {\L}ukasiewicz logic in the scope of $\cdia$-connectives.

Next we define truth of formulas on deliberative Kripke models. We begin by giving the valuation of complex formulas from $\lblack$, which is simply three-valued {\L}ukasiewicz logic.

\begin{definition}[$\alpha$-satisfaction] For any three-partitioning $\pi = (\pi_1,\pi_0,\pi_{\frac{1}{2}})$ of $\Pi$, we define 
\begin{align*}
\overline \pi(p) &= x \text{ s.t } p \in \pi_x \\
\overline\pi(\neg\alpha) &= 1 - \overline\pi(\alpha)\\
\overline\pi(\alpha_1 \to \alpha_2) &= \min \{1, 1 - (\overline\pi(\alpha_1) - \overline\pi(\alpha_2))\} 
\end{align*}
\end{definition}

Now we can give a semantic interpretation of the full language as follows.

\begin{definition}[$\dlangm$-satisfaction]\label{def:ddlm}
Given an argumentation semantics $\sem$ and a basis $\views$, truth on $\kmod \views \sem$ is defined inductively as follows, in all points $q \in \carriers \views$.
\begin{align*}
\kmod \views \sem, q \vDash \cdia \alpha & \iff \quad \text{there is } \pi \in \pis \sem(q) \text{ s.t. } \overline\pi(\phi) = 1 \\
\kmod \views \sem, q \vDash \neg \phi \quad & \iff \quad \text{not } \kmod \views \sem, q \vDash \phi \\
\kmod \views \sem, q \vDash \phi \wedge \psi & \iff \quad  \text{both } \kmod \views \sem, q \vDash \phi \text{ and } \kmod \views \sem, q \vDash \psi \\
\kmod \views \sem, q \vDash \ddia \pi p & \iff \quad \exists (q,q') \in \rels \views(p): \kmod \views \sem,q' \vDash \phi \\
\kmod \views \sem, q \vDash \adia \phi & \iff \quad \exists (q,q') \in \rels \views(\exists): \kmod \views \sem,q' \vDash \phi \\
\end{align*}
\end{definition}


To illustrate the definition, we return to the example depicted in (\ref{ex}). In Figure \ref{tbl:dynamism-3}, we depict this basis together with a fragment of the corresponding Kripke model, in particular the fragment arising from the $p$-successors of $(\emptyset,\emptyset)$.

%

\begin{figure}[ht]
\begin{minipage}[t]{10em}
$$ \small
\views = \left\{\begin{array}{l} V_a: \hspace{1.5em} \xymatrix{p \ar@(lu,ld) \ar@/_/[r] & q \ar@/_/[l] } \vspace{2em}  \\ V_b: \xymatrix{p \ar@/_/[r] & q \ar@(ru,rd) \ar@/_/[l] }
\end{array}\right\}$$
\end{minipage}
\begin{minipage}[t]{28em} \hspace{0.2em}
  \centering
  \def\svgwidth{\textwidth}
  \begin{tiny}
    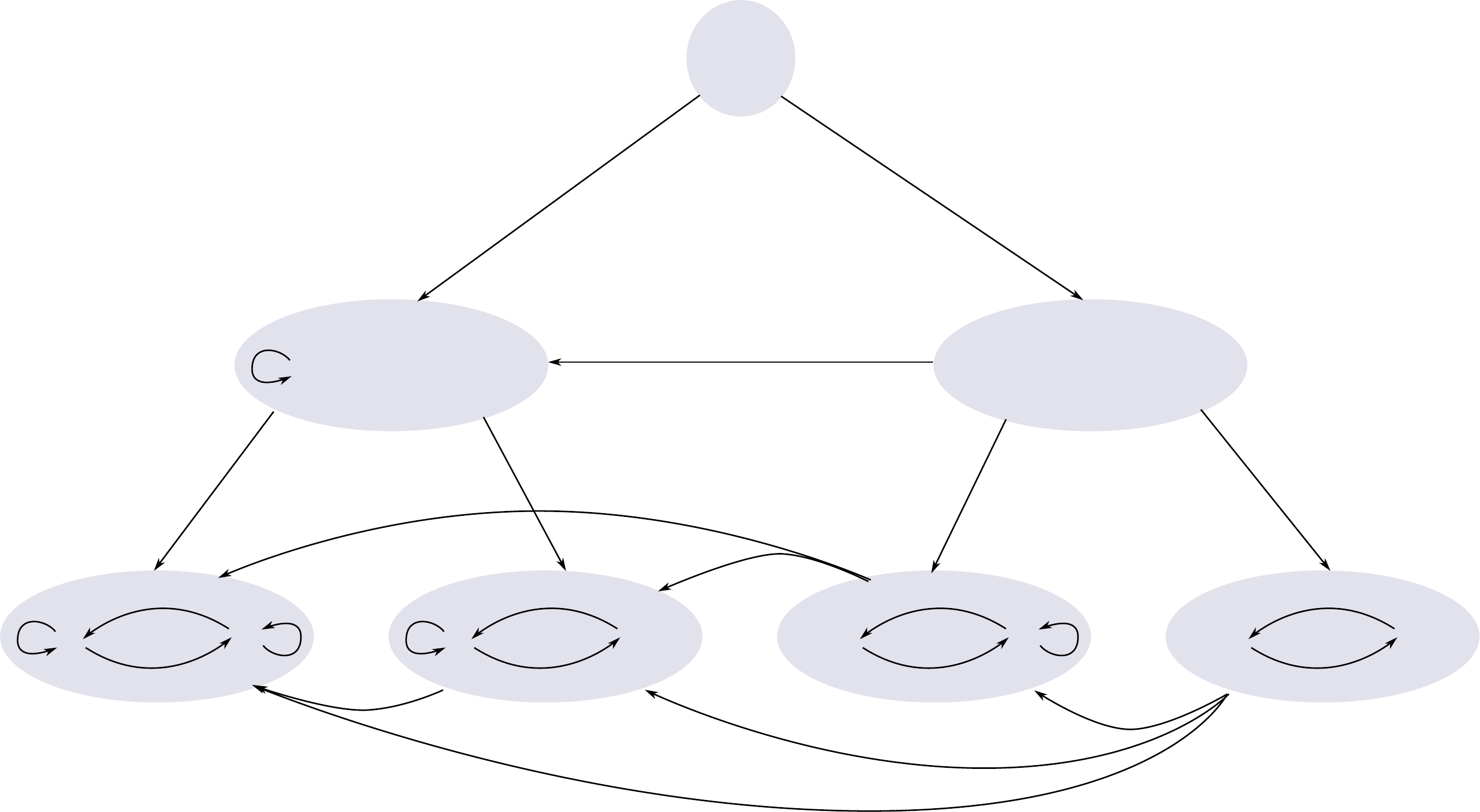
  \end{tiny}
\end{minipage}
  \caption{A fragment of the deliberative Kripke model for $\views$.}
  \label{tbl:dynamism-3}
\end{figure}

Let us assume that $\sem = \sf p$ is the preferred semantics. Then the following list gives some formulas that are true on $\kmod \views \sem$ at the point $(\emptyset,\emptyset)$, and the reader should easily be able to verify them by consulting the above fragment of $\kmod \views \sem$.

$$
\begin{array}{lll}
\ddia p \cbox p, & \adia \cbox p, & [p] \adia \cbox q, \\
\neg [p] \adia \cdia p, & \ddia p \adia \cbox \neg p, & \adia \adia (\cdia p \land \cdia q)
\end{array}
$$

We can also record some validities that are easy to verify against Definition \ref{def:main}. 

\begin{proposition}\label{prop:val}
The following formulas are all validities of $\dlangm$, for any $p,q \in \Pi$, $\phi \in \dlangm$. \begin{enumerate}
\item $\ddia p \ddia q \phi \leftrightarrow \ddia q \ddia p \phi$
\item $\ddia p \dbox q \phi \to \dbox q \ddia p \phi$
\item $\adia \abox \phi \to \abox \adia \phi$
\item $\ddia p \ddia p \phi \to \ddia p \phi$ 
\end{enumerate}
\end{proposition}

We remark that $\dbox q \ddia p \phi \to \ddia p \dbox q \phi$ is \emph{not} valid, as witnessed for instance by the following basis $\views$, for which we have $\kmod \views p, (\emptyset,\emptyset) \models \dbox q \ddia p \cbox p$ but also $\kmod \views p, (\emptyset,\emptyset) \models \dbox p \ddia q \cbox q$ (as the reader may easily verify by considering the corresponding Kripke model).

$$
\views = \left\{\begin{array}{ll}V_a: \xymatrix{p \ar[r] & q} \\ V_b: \xymatrix{p & q \ar[l]} \end{array}\right\}$$

Finally, let us notice that as $\Pi$ is generally infinite, we must expect to encounter infinite bases. This means, in particular, that our Kripke models are often infinite. However, in the next section we show that as long as $\views$ is \emph{finitary}, meaning that no agent $a \in \agents$ has a view where an argument is attacked by infinitely many other arguments, we can solve the model-checking problem also on infinite models.

\section{Model checking on finitary models}\label{sec:mcheck}

Towards this result, we now introduce some notation and a few abstractions to simplify our further arguments. We will work with labeled trees, in particular, where we take a tree over labels $X$ to be some
 non-empty, prefix-closed subset of $X^*$ (finite sequences of elements of $X$). Notice that trees thus defined contain no infinite sequences. This is intentional, since we will ``shrink'' our models (which may contain infinite sequences of related points), by mapping them to trees. To this end we will use the following structures. 

\begin{definition}\label{def:iaf} Given a basis $\views$, we define $\iterate \views$, a set of sequences over $\Pi \times 2^\Pi$ labeled by AFs, defined inductively as follows
\begin{description}
\item [Base case:] $\epsilon \in \iterate \views$ and is labeled by the AF $\af(\epsilon) = (S(\epsilon),E(\epsilon))$ where $S(\epsilon) = \emptyset = E(\epsilon)$.
\item [Induction step:] If $x \in \iterate \views$, then for any $p \in \Pi$ and any partial assent $X = \update \views x p$, we have $x;(p, X) \in \iterate \views$ labeled by the AF $\af(x;(p,X))$ where $S(x;(p, X)) = S(x) \cup \{p\}$ and $E(x;(p, X)) = E(x) \cup X$.
\end{description}
\end{definition}
To adhere to standard naming we use $\epsilon$ to denote the empty string. It should not be confused with the argumentation semantics $\sem$. This will also be clear from the context.
We next define tree-representations of our Kripke models.

\begin{definition}\label{def:treerep} Let $\kmod \views \sem$ be some model. The \emph{tree representation} of $\kmod \views \sem$ is the set $T$, together with the representation map $\gamma: \carriers \views \to 2^T$, defined inductively as follows
\begin{description}
\item[Base case] $\epsilon \in T$ is the root with $\gamma((\emptyset, \emptyset)) = \{\epsilon\}$.
\item[Induction step] For any $x \in T, q \in \kmod \views \sem$ with $x \in \gamma(q)$ and $q' \in succ(q)$ witnessed by $p \in \Pi$ and $X \in \update \views x p$, we have $x;(p, X) \in T$ with $q' \in \gamma(x;(p,X))$.
\end{description}
\end{definition}

Notice that the tree-representation is a tree where each node is an element of $\iterate \views$. Some single states in $\kmod \views \sem$ will have several representations in a tree. That is, $\gamma(q)$ may not be a singleton. On the other hand, it is easy to see that for every state $q \in \kmod \views \sem$, and every path from $(\emptyset, \emptyset)$ to $q$, there will be a node $x \in T$ such that $q \in \gamma(x)$.

The main result of our paper is that model checking $\dlangm$-truth at $(\emptyset,\emptyset)$ is tractable as long as all views are \emph{finitely branching}, i.e., such that for all $a \in \agents, p \in \Pi$, $p$ has only finitely many attackers in $V_a$. Clearly this requires shrinking the models since the modality $\adia$ quantifies over an infinite domain whenever $\Pi$ is infinite. We show, however, that attention can be restricted to arguments from $\Pi$ that are \emph{relevant} to the formula we are considering. To make the notion of relevance formal, we will need the following measure of complexity of formulas.
%

\begin{definition}\label{def:depth} The \emph{white modal depth} of $\phi \in \dlangm$ is $\depth{\phi} \in \mathbb N$, which is defined inductively as follows 
\begin{align*}
\depth{\alpha} \quad & := \quad 0 & \text{no white connectives in these formulas}\\
\depth{\cdia \alpha} \quad & := \quad 0 \\
\depth{\neg \phi} \quad &:= \quad \depth{\phi} & \text{depth is deepest nesting of }\\
\depth{\phi \wedge \psi} \quad &:= \quad \max\{\depth{\phi}, \depth{\psi}\} & \text{white connectives}\\
\depth{\adia \phi} \quad &:= \quad 1 + \depth{\phi} \\
\depth{\ddia p \phi} \quad &:= \quad 1 + \depth{\phi} \\
\end{align*}
\end{definition}

We let $\restr \Pi \phi$ denote the set of arguments occurring in $\phi$ in sub-formulas from $\lblack$. Notice that given a state $q \in \carriers \views$, the satisfaction of a formula of the form  $\phi = \cdia \alpha$ at the AF encoded by $q$ is not dependent on the entire digraph $q = (q_S,q_E)$.

Indeed, this is what motivated our definition of normality for an argumentation semantics, leading to the following simple lemma, which is the first step towards shrinking Kripke structures for the purpose of model checking. Given a model $\kmod \views \sem$ and a state $q \in \carriers \views$, we let $\comp q \Phi$ denote the digraph consisting of all connected components from $q$ which contains a symbol from $\Phi$. Then we obtain the following.

\begin{lemma}\label{lemma:comp}Given a semantics $\sem$ and two bases $\views$ and $\views'$, we have, for any two states $q \in \kmod \views \sem$ and $q' \in \kmod {\views'} \sem$ and for any formula $\phi \in \dlangm$ with $\depth \phi = 0$:
$$
\big(\comp q {\restr \Pi \phi} = \comp {q'} {\restr \Pi \phi}\big) \Rightarrow \big(\kmod \views \sem,q \vDash \phi \Leftrightarrow \kmod {\views'} \sem, q'  \vDash \phi \big)$$
\end{lemma}

In order to complete our argument in this section, we will make use of $n$-bisimulations modulo a set of symbols.

\begin{definition}\label{def:bisim} Given two models (with possibly different bases, but with common set of symbols $\Pi$ and semantic $\sem$) $K_\views = \langle \carriers \views, R, \sem \rangle$ and $K_\views' = \langle \carriers {\views'}, R', \sem\rangle$, states $q \in \carrier$ and $q' \in \carrier'$, a natural number $n$ and a set $\Phi \subseteq \Pi$, then we say that $q$ and $q'$ are $n$-bisimilar modulo $\Phi$ (denoted $(K_\views, q) ~ \bisim_n^\Phi~ (K_{\views'}, q')$), if, and only if, there are $n+1$ relations relation $Z_n \subseteq Z_{n-1} \subseteq \dots \subseteq Z_0 \subseteq \carrier \times \carrier'$ such that
\begin{enumerate}
\item $q Z_n q'$, 
\item whenever $(v, v') \in Z_0$, then $C(v, \Phi) = C(v', \Phi)$, 
\item whenever $(v, v') \in Z_{i+1}$ and $vRu$, then there is a $u'$ s.t. $v'R'u'$ and $uZ_{i}u'$, 
\item whenever $(v, v') \in Z_{i+1}$ and $v'R'u'$, then there is a $u$ s.t. $vRu$ and $uZ_i u'$.
\end{enumerate}
\end{definition}

Let us now also define a particular subset of arguments, the arguments which have at most distance $n$ from some given set of arguments: 

\begin{definition}\label{def:vicinity}Given a \basis $\views = \viewsv$, a subset $\Phi \subseteq \Pi$ and a number $n$, the $n$-vicinity of $\Phi$ is $D(\views, \Phi, i) \subseteq \Pi$, defined inductively as follows
\begin{align*}
D(\views, \Phi, 0) &= \Phi \\
D(\views, \Phi, n+1) &= D(\views, \Phi, n) \\ 
&\cup \left\{~p \in \Pi ~|~ \exists q \in D(\views, \Phi, n): \{(p,q),(q,p)\} \cap \bigcup_{a \in \agents}V_a \not = \emptyset ~\right\}\\
\end{align*}
\end{definition}

Notice that as long as $\Phi$ is finite and all agents' views have finite branching, then the set $D$ is also finite. Also notice that an equivalent characterization of the set $D(\viewsv,\Phi,i)$ can be given in terms of paths as follows: an argument $p \in \Pi$ is in $D(\views,\Phi,i)$ if, and only if, there is a path $p=x_1x_2\ldots x_n$ in $\bigcup_{a \in \agents}V_a$ such that $x_n \in \Phi$ and $n \leq i$ (we consider an argument $p$ equivalently as an empty path at $p$). 

\begin{definition}\label{def:shrink} Given a formula $\phi \in \dlangm$. Let $\viewsv$ be a possibly infinite \basis, we define $\shrink \phi \viewsv$ such that
\begin{itemize}
\item for every $a \in \agents$, $\shrink \phi {V_a} ~:=~ V_a \cap D(V_a, \restr \Pi \phi, \depth{\phi})$
\end{itemize}
\end{definition}

Notice that the Kripke model for $\rho(\views)$ will have finite branching as long as the argument symbols in the $\depth{\phi}$-vicinity of the argument symbols in $\phi$ have finite branching in all agents' views. In the following, we
 will show that for any finitely branching $\views$ and normal $\sem$, we have $\kmod \views \sem, (\emptyset,\emptyset) \models \phi$ if, and only if, $\kmod {\shrink \phi \views} \sem, (\emptyset,\emptyset) \models \phi$. 

\begin{theorem} Let $\views$ be an arbitrary \basis, and $\phi \in \dlangm$. 
$$ \left(\kmod \views \sem, (\emptyset,\emptyset)\right) \quad   \bisim_{\depth{\phi}}^{\restr \Pi \phi}  \quad \left(\kmod {\shrink \phi \views} \sem, (\emptyset,\emptyset)\right)$$
\end{theorem}


\begin{proof}
Let $\kmod \views \sem$ be an arbitrary model and let $T$ denote its tree representation, while $T'$ denotes the tree representation of $\kmod {\rho_\phi(\views)} \sem$.

We take $n = \depth{\phi}$ and let $\Phi$ be the atoms occurring in $\phi$ inside the scope of some $\cdia$-operator. Moreover, for brevity, we denote $D = D(\views, \Phi, n)$. 

\paragraph{Definition of $(Z_i)_{(0 \leq i \leq \depth{\phi})}$: } 
We define all the relations $Z_i$ inductively using the tree-representations as follows.
\begin{description}
\item[Base case: ] ($i = 0$) For all $0 \leq i \leq n$, we let $\epsilon Z_i \epsilon$. 
\item[Induction step: ] ($0 < i \leq n$) For all $y = x;(v, X) \in T$ and $y' = x';(v', X') \in T'$, both of length $i$, with $x (Z_{i+1}) x'$. We let, for every $k \leq i$, $y (Z_k) y'$ if, and only if, $v = v'$, \emph{and} $X \cap (D \times D)= X'$.
\end{description}

Notice that if $x (Z_i) x'$, then $S(x) = S(x')$ and $|S(x)| \leq (n - i)$. Moreover, by consulting Definition \ref{def:treerep} it is not hard to see that for all $q \in  \carriers \views , q' \in \carriers {\rho(\views)}$ we have, for all $0 \leq i \leq n$ and all $q \in \carriers \views,q' \in \carriers {\rho(\views)}$: 
$$
\forall x_1,x_2 \in \gamma(q): \forall x'_1,x'_2 \in \gamma(q'): x_1(Z_i)x_2 \iff x'_1 (Z_i) x'_2 
$$
This means, in particular, that the following lifting of $(Z_i)_{0 \leq i \leq n}$ to models is well-defined, for all $q \in \carriers \views, q' \in \carriers {\rho(\views)}$ and all $0 \leq i \leq n$:
$$
q(Z_i)q' \iff x(Z_i)x' 
$$
for some $x \in \gamma(x), x' \in \gamma(q')$.

Next we show that $(Z_i)_{0 \leq i \leq n}$ so defined is an n-bisimulation between $\kmod \views \sem$ and $\kmod {\rho(\views)} \sem$.

\paragraph{$(Z_i)_{0 \leq i \leq n}$ witnesses $n$-bisimulation:} We address all the points of the definition of $n$-bisimulation modulo $\Phi$ in order.
\begin{enumerate}
\item Clearly, $(\emptyset, \emptyset) Z_n (\emptyset, \emptyset)$. Hence the first condition of the definition is satisfied. 
\item Consider any arbitrary states $q,q'$ and let $x = x_1;x_2;\dots;x_m$ and $x' = x_1';x_2';\dots;x_m'$ be the corresponding nodes from $T,T'$ that witnesses to $q(Z_0)q'$. By definition of $Z_0$ we have $S(x) = S(x')$, but it is possible that we have $E(x) \not = E(x')$. However, we must have $C(\af(x), \Phi) = C(\af(x'), \Phi)$, and to see this, it is enough to observe that as $m \leq n$, each of $x$ and $x'$ contains at most $n$ nodes. Then, since $\af(x) = q$ and $\af(x') = q'$ are the same on $D$, and the distance from $\Pi \setminus D$ to $\Phi$ is greater than $n$. That is, any path from an argument in $\Pi \setminus D$ to an argument in $\Phi = \restr \Pi \phi$ would be a path consisting of at least $n+1$ nodes. It follows that no element from $\Phi$ can be in a connected components containing elements outside of $D$.
\item Consider now $q,q'$ corresponding to $x$ and $x'$ such that $x (Z_{i + 1}) x'$. Notice that $(q,r) \in \rels \views(\exists)$ if, and only if, there is a $(p, X)$ such that $x R (x;(p, X))$. So all we need to show is that $X \cap (D \times D)$ is in $\update {\rho_\phi(\views)} {x'} p$. Then it will follow that there is a successor to $x'$, namely $(p, X \cap (D \times D))$, with $(x')R'(x';(p, X \cap (D \times D)))$. This is a straightforward consequence of the Definition \ref{def:shrink} of $\rho$. The argument for the particular sub relations $\rels \views(p)$ is analogous. 
\item Finally consider $q,q'$ corresponding to $x$ and $x'$ such that $x (Z_{i + 1}) x'$ for $(p, X)$ such that $x' R (x';(p, X'))$. Again we need to ensure that there is an $X \in \update \views x p$ such that $X' = X \cap (D \times D)$, and again this follows from the Definition \ref{def:shrink} of $\rho$. The argument for the particular sub relations $\rels \views(p)$ is analogous. 
\end{enumerate}
\end{proof}

\begin{proposition} Let $\phi \in \dlangm$ and $\views, \views'$ arbitrary bases. If states $q \in \kmod \views \sem$ and $q' \in \kmod {\views'} \sem$ are $\depth{\phi}$-bisimilar modulo $\restr \Pi \phi$, then $\kmod \views \sem, q \models \phi ~\Leftrightarrow~ \kmod {\views'} \sem , q' \models \phi$. Or, succinctly
$$ \left((\kmod \views \sem, q) ~\bisim_{\depth{\phi}}^{\restr \Pi \phi}~ (\kmod {\views'} \sem, q')\right) ~\Rightarrow~ \left(\kmod \views \sem, q \models \phi ~\Leftrightarrow~ \kmod {\views'} \sem, q' \models \phi\right).$$
\end{proposition}

\begin{proof} The proof is by induction on $\depth{\phi}$. 
\begin{description}
\item[Base case:] ($\depth{\phi} = 0$) There are no white connectives, and our states, $q$ and $q'$, are clearly 0-bisimilar modulo $\Phi$. It is also easy to see, consulting Definition \ref{def:main}, that the truth of a formula of modal depth $0$ is only dependent on the AF $q$. Then it follows from the fact that $\sem$ is assumed to be normal that the truth of $\phi$ is in fact only dependent on $C(q, \Phi)$. From $q (Z_0) q'$, we obtain $C(q, \Phi) = C(q', \Phi)$ and the claim follows. 
\item[Induction step:] ($\depth{\phi} > 0$) We skip the boolean cases as these are trivial, so let $\phi := \adia \psi$ (the case of white connectives with an explicit argument is similar). Suppose $\depth{\phi} = i + 1$ and $q (Z_{i+1}) q'$. Suppose further that $\kmod \views \sem, q \models \adia \psi$. Then there is a successor of $q$, $v \in succ(q)$ such that $\kmod \views \sem, v \models \psi$. All successors of $q$ will be $i$-bisimilar to a successor of $q'$ (point 3. of Definition \ref{def:bisim}). So we have $(\kmod \views \sem, v)~\bisim_i^{\Phi}~(\kmod {\views'} \sem, v')$. As $\depth{\psi} < \depth{\adia \psi}$ we can apply our induction hypothesis to obtain $\kmod {\views'} \sem, v' \models \psi$, and $\kmod {\views} \sem, q' \models \adia \psi$ as desired.
\end{description}

\end{proof}

\section{Conclusion and future work}\label{sec:concfut}

We have argued for a logical analysis of deliberative processes by way of modal logic, where we avoid making restrictions that may not be generally applicable, and instead focus on logical analysis of the space of possible outcomes. The deliberative dynamic logic (\acro{ddl}) was put forth as a concrete proposal, and we showed some results on model checking.

We notice that \acro{ddl} only allows us to study deliberative processes where every step in the process is explicitly mentioned in the formula. That is, while we quantify over the arguments involved and the way in which updates take place, we do not quantify over the \emph{depth} of the update. For instance, a formula like $\Diamond \Box p$ reads that there is a deliberative update such that no matter what update we perform next, we get $\phi$. A natural next step is to consider instead a formula $\Diamond \Box^\ast \phi$, with the intended reading that there is an update which not only makes $\phi$ true, but ensures that it remains true for all possible future \emph{sequences} of updates. Introducing such formulas to the logic, allowing the deliberative modalities to be iterated, is an important challenge for future work.  Moreover, we would also like to consider even more complex temporal operators, such as those of computational tree logic, or even $\mu$-calculus. 

Finding finite representations for the deliberative truths that can be expressed in such languages appears to be much more challenging, but we would like to explore the possibility of doing so.

Also, we would like to explore the question of validity for the resulting logics, and the possibility of obtaining some compactness results. Indeed, it seems that if we introduce temporal operators we will be able to express truths on arbitrary points $q \in \carriers \views$ by corresponding formulas that are true at $(\emptyset,\emptyset)$, thus capturing the way in which complete assent can be faithfully captured by a finite (albeit unbounded) notion of iterated deliberation. 

If the history of the human race is anything to go by, it seems clear that we never run out of arguments or controversy. But it might also be that some patterns or structures are decisive enough that they warrant us to conclude that the \emph{truth} has been settled, even if deliberation may go on indefinitely. A further logical inquiry into this and related questions will be investigated in future work.

\bibliographystyle{abbrv}
\bibliography{cites,more_cites}
\end{document}